\documentclass[review,12pt]{elsarticle}
\usepackage{amsmath}
\usepackage{amssymb}
\usepackage{amsthm}
\theoremstyle{plain}
\newtheorem{thm}{Theorem}[section]
\theoremstyle{definition}
\newtheorem{defn}{Definition}[section]

\newcommand{\Ker}{\operatorname{Ker}}

\begin{document}
\begin{frontmatter}

\title{On decoding procedures of intertwining codes}



%
\author{Shyambhu Mukherjee}
\ead{pakummukherjee@gmail.com}
\address{SMU Department,
         Indian Statistical Institute\\ 
         Bangalore, Karnataka, India.}

\author{Joydeb Pal}
\ead{joydebpal77@gmail.com}
\address{Department of Mathematics\\ 
         National Institute of Technology Durgapur\\
         Burdwan, India. \\
         }

\author{Satya Bagchi} 
\ead{satya.bagchi@maths.nitdgp.ac.in}
\address{Department of Mathematics\\ 
         National Institute of Technology Durgapur\\
         Burdwan, India.}

\begin{abstract}
One of the main weakness of the family of centralizer codes is that its length is always $n^2$. Thus we have taken a new matrix equation code called intertwining code. Specialty of this code is the length of it, which is of the form $n k$. We establish two decoding methods which can be fitted to intertwining codes as well as for any linear codes. We also show an inclusion of linear codes into a special class of intertwining codes.
\end{abstract}
\end{frontmatter}
\newpage
\section{Introduction}
A code of length $n^2$ is obtained by taking centralizer of a matrix from the vector space $\mathbb{F}_q^{n \times n}$. As a consequence, it cannot reach to most of the sizes. Whereas a code of length $n \cdot k$ is formed by taking the solutions of matrix equation for some matrices $A \in \mathbb{F}_q^{n \times n}$ and $C \in \mathbb{F}_q^{k \times k}$ over $\mathbb{F}_q$. Here we have taken one such matrix equation of the form $AB=BC$ where the matrix $A  \in \mathbb{F}_q^{n \times n}$  and the matrix $C \in \mathbb{F}_q^{k \times k}$. Thus the set of solutions $\mathcal{R}(A,C)= \{B \in \mathbb{F}_q^{n \times k}|AB=BC\}$ gives a code of length $n\cdot k$ by a similar construction in \cite{Alahmadi2017235} \cite{Alahmadi201468}  \cite{joydeb1}. This code is named intertwining code \cite{Glasby2107}. This code can extend the use of better decoding ability of GTC codes \cite{joydeb1} into a vast class of linear codes.

Finding efficient error correcting procedure for a linear code is a challenging problem. If we look upon centralizer codes and twisted centralizer codes, we see that there was a very nice method to detect and correct single error using syndrome. This technique cannot provide an easy task for correcting more than single errors. Thus we present two algorithms to do better decoding procedure. Our algorithms work for the family of intertwining codes as well as for any linear codes.

In this paper, we explore few properties on intertwining codes. Then we show that there exists a intertwining code for which a certain linear code is a subcode of it. We find a way to effectively find an upper bound on the minimum distance of the code along with proving the existence of a certain weight codeword based upon the matrices $A$ and $C$. At last we show a possible way to find the matrix pair $(A,C)$ of intertwining code related to a linear code in total computational perspective.

Throughout this paper we denote $\mathbb{F}_q$ as a finite field with $q$ elements and $\mathbb{F}_q^{n \times k}$ as the set of all matrices of order $n \times k$ over $\mathbb{F}_q$. We take two matrices $A$ and $C$ from the vector spaces $\mathbb{F}_q^{n \times n}$ and $\mathbb{F}_q^{k \times k}$ respectively and also $O$ denotes the null matrix.

\begin{defn}
 For any matrices  $A \in \mathbb{F}_q^{n \times n}$ and $C \in \mathbb{F}_q^{k \times k}$, the set $\mathcal{R}(A, C)= \lbrace B \in \mathbb{F}_q^{n \times k}| AB=BC \rbrace$ is called intertwining code \cite{Glasby2107}. 
\end{defn}

Clearly, the set $\mathcal{R}(A, C)$ is a linear subspace of the vector space $\mathbb{F}_q^{n \times n}$ and hence it is a linear code. The formation of the code is very similar to \cite{Alahmadi2017235} \cite{Alahmadi201468} \cite{joydeb1}. We develop few basic results on intertwining codes as follows.


\begin{enumerate}

\item If $A \in \mathcal{R} (A,C)$ then $C$ should be an $n \times n$ matrix. Now let $A$ belongs to the intertwining code then $A^2=AC$. If $A$ is idempotent then $A(C-I_n) = O$ and if $A^2=0$  then $C$ belongs to the matrix wise null space of $A$. But these are not possibly the whole solution of the equation $A^2=AC$.
\item If $A$ and $C$ are invertible matrices of order $n$ and $k$ respectively, then $\mathcal{R}(A, C)$ is isomorphic to $\mathcal{R}(A^{-1}, C^{-1})$.
\item Let $X \in \mathbb{F}_q^{n \times n}$  and $Y \in \mathbb{F}_q^{k \times k}$ are two invertible matrices. Then $\mathcal{R}(A,C)$ is isomorphic to conjugate code $\mathcal{R}(XAX^{-1}, YCY^{-1})$.
\item Let $A$ is such a matrix that $c$ is not an eigenvalue and $C =cI_k$. Then the matrix equation $AB=BC$ possesses a trivial solution, i.e., $\mathcal{R}(A, C)=\{0\}$.
\end{enumerate}


\section{Analysis on weight distribution}
\begin{thm}
Let $J \in \mathbb{F}_2^{n \times k}$ be the matrix with all entries are $1$. If $J \in \mathcal{R}(A,C)$ then the weight distribution of the code $\mathcal{R}(A,C)$ is symmetric, i.e., $A_i=A_{nk-i}$ for all $i$. In addition, if the matrix $A$ is of row sum equal to $0$ and $C$ is of column sum equal to $0$ then $J \in \mathcal{R}(A,C)$.
\end{thm}

\begin{proof}
 Let $B_1 \in \mathcal{R}(A,C)$ with weight $i$ where $0 \leq i \leq nk$. If $J \in \mathcal{R}(A,C)$ then $J +B_1 \in \mathcal{R}(A,C)$, since $\mathcal{R}(A,C)$ is a linear code. Now weight of $ J +B_1$ is $nk-i$. So, whenever a codeword of weight $i$ appears, simultaneously there exists a codeword with weight $nk-i$. Thus we can say weight distribution is symmetric, i.e., number of codewords with weight $i$ = number of codewords with weight $nk-i$, where $i$ is a positive integer with $0 \leq i \leq nk$, i.e., $A_i=A_{nk-i}$.

For the second part, observe that $AJ = [ row$-$1$-$sum ~row$-$2$-$sum ~\dots ~row$-$n$-$sum]^T$ $\cdot$ $[1 ~1 ~1 ~1 ~\dots ~1]$ and similarly observe that $JC=[1 ~1 ~1 ~\dots ~1]^T \cdot [col$-$1$-$sum ~col$-$2$-$sum$ $ ~\dots ~col$-$k$-$sum]$. Hence it is clear that if the matrices $A$ has row sum equal to $0$ and $C$ has column sum equal to $0$, then $J$ satisfies the equation $AJ-JC=O$ since  $AJ=O$ and $JC=O$. 
\end{proof}

\subsection{Existence of a certain weight codeword}
Consider the matrices $A$ and $C$ for which $\mathcal{R}(A,C)$ has been constructed. Let us assume that none of them are invertible. Therefore, there will be dependent columns and rows. If $d_A$ is the minimum number of columns which are dependent in $A$ and $d_C$ is the minimum number of rows which are dependent in  $C$, then at least one codeword in $\mathcal{R}(A,C)$ of weight $d_A \cdot d_C$ exists. In this case, we have $\sum c_i A_i=0$ with $c_i \neq 0, ~\forall i$ as the relation is taken to be of minimum number of columns. Then we have a column vector of those $c_i$'s. Similarly,  we get a row vector from $C$. Then we take the multiplication of these vectors as $n\times 1$ into $1 \times k$ matrix multiplication. So this vector will give us an element of $\mathcal{T}_O=\{ B \in \mathbb{F}_q^{n \times k}| AB=BC=O\}$.  
 
As the codeword of weight $d_A \cdot d_C$ exists, by construction this weight $d_A \cdot d_C$ is less than or equal to the $(r_A+1)\cdot(r_C+1)$ hence the minimum weight is also less than or equal to $(r_A+1) \cdot(r_{C}+1)$. We list the result in the above discussion as follows.

\begin{thm} 
For an intertwining code generated by $A$ and $C$, there exists at least one codeword of weight $d_A \cdot d_C$ and therefore the minimum distance $d$ of $\mathcal{R}(A,C)$ is less than or equal to $d_A \cdot d_C$, i. e., $d(\mathcal{R}(A,C)) \leq d_A\cdot d_C$.
\end{thm}

\begin{proof}
The product code  $\Ker A \otimes \Ker C^T$ belongs to intertwining code. Now according to definition of $d_A$ there exists $d_A$ number of columns dependent such that no set of lesser cardinality is dependent. Therefore a column vector $v$ of weight $d_A$ exists such that $Av=0$. Similarly we find a row vector $w$ of weight $d_C$ with $wC=0$ therefore $vw$ exists in the code with the rest following.
\end{proof}

\section{Decoding process}
Encoding procedure for intertwining codes are similar to the encoding procedures mentioned in \cite{Alahmadi2017235} \cite{Alahmadi201468} \cite{joydeb1}. To check whether a message is erroneous, it is required to define syndrome.

\begin{defn}
Let $A$ be a square matrix of order $n$ and $C$ be a square matrix of order $k$. Then the syndrome of an element $B \in \mathbb{F}_q^{n \times k}$ with respect to the intertwining code $\mathcal{R}(A,C)$ is defined as $S_{A,C}(B)= AB-BC$.
\end{defn}

If a word is erroneous then $S_{A,C}(B)= AB-BC \neq O$. It is an easiest technique to check whether a codeword belongs to the code or not. Here we propose two algorithms to correct errors in an intertwining code.

\subsection{\textbf{Algorithm 1:}}
As we do not know how to find the minimum weight element algorithmically hence we can not keep it inside our decoding process as it needs a lot of time. That's why we modify our procedure a bit for the Step $3$. As the vector space $\mathbb{F}_q^{n \times k}$ can be broken down into intertwining code and its cosets so we will calculate minimum weight element first for all the cosets separately. Now we input list of coset leaders as a table inside our decoding system.

\begin{enumerate}
\item[Step $1.$] Receiver received a word $B'$ from the channel.
\item[Step $2.$] Calculate the syndrome $S_{A,C}(B')=AB'-B'C$. If  $S_{A,C}(B')=O$ then the transmitted codeword is $B'$ and goto Step $5$. Otherwise, goto Step $3$.
\item[Step $3.$] Find this syndrome’s corresponding least weight error matrix $E$, already stored in the table. The matrix $E$ is the error pattern for the word $B'$.
\item[Step $4.$] Since, both $B'$ and the $E$ belong to the same coset $B'+\mathcal{R}(A,C)$, then $B'-E$ is the transmitted codeword of the code $\mathcal{R}(A,C)$. 
\item[Step $5.$] End.
\end{enumerate}

\noindent \textbf{Caution!} This table may look like syndrome look-up table and it will work similarly, but here we must remember following differences. 

\begin{enumerate}
\item This syndrome is different from the standard syndrome which is obtained by multiplying the codeword of length $n$ with a $(n-k) \times n$ parity-check matrix, resulting into $n-k$ length i.e. $7-4=3$ length for hamming code. Here we will get a whole matrix of length $n^2$, same length as our codewords, as the minimum weight codeword of a coset.
\item This table is not the syndrome look-up table for the intertwining code. A different one can be constructed but the construction will be more complex if we try to do so.
\end{enumerate}

\subsection{\textbf{Algorithm 2:}}
In the previous algorithm, we have to store whole of a table which occupy a good amount of memory. So, we provide a better algorithm which does not take the memory that much and achieves the least weight error matrix or error pattern by this algorithm. 

For this algorithm, partition the codewords of the code $\mathcal{R}(A,C)$ by its weight distribution. Let $\mathcal{R}(A,C) = \cup_{i=1}^k A_k$, where $A_j = \lbrace c \in \mathcal{R}(A,C): wt(c)=a_j \rbrace$. We easily see that $A_i \cap A_j = \phi, ~\forall ~i \neq j$. Now weights available are $a_0, a_1, \dots, a_k$. Now, the algorithm is presented below.

\begin{enumerate}
\item[Step $1.$] Receiver received a word $B'$ and start to calculate its syndrome $S_{A,C}(B')=AB'-B'C$. To reduce complexity for computing the syndrome, we calculate bitwise syndrome. Whenever a nonzero bit appears in the syndrome, we stop the computation of the syndrome and goto Step $2$. Otherwise if $S_{A,C}(B')= 0$ then the transmitted word is $B'$ and goto Step $5$.
\item[Step $2.$] Find $b=wt(B')$, weight of $B'$. Now evaluate the intervals in two cases. If $n \geq b+a_i$, then we take the interval $[|b-a_i|, b+a_i]$, otherwise take the interval $[|b-a_i|, 2n-b-a_i]$.
\item[Step $3.$] Delete those intervals whose lower bounds are greater than $t$, where $t = \lfloor \frac{d-1}{2} \rfloor$ and $d$ is the minimum distance of the code $\mathcal{R}(A,C)$. 
\item[Step $4.$] Arrange remaining intervals in ascending order of the lower bounds of intervals. We store ordering of index. Take the first interval then take corresponding weight. Let $a_m$ be the corresponding weight. Now find $S = \{A+B': wt(A+B') \leq  wt(A'+B') ~\forall ~A' \in A_m\}$. Select $A$ such that $wt(A+B')$ is minimum for $A \in A_m$. If not unique, choose one $A$ randomly. Let $E=A+B'$. Find weight of $E$ and then delete intervals with lower bound greater than or equal to $wt(E)$. Go to next partition. In same procedure find the least weight word $E'$. Compare with the previous least weight word $E$. Choose minimum weight among these and save it to $E$. Continuing this process for further partitions we will get a matrix $E$, which is the error pattern. Therefore the transmitted codeword was $B'-E$.
\item[Step $5.$] End.
\end{enumerate}

\textbf{Note:} In Step $4$ of the above algorithm, error pattern matrix $E$ is always unique because if there are $E_1$ and $E_2$ such that both $B'-E_1$ and $B'-E_2$ have weight less than $t$ then distance between $E_1$ and $E_2$ will be less than $2t<d$ which is impossible for two codewords.

\textbf{Analysis:} In our algorithm we will store the code sorted according to weights. So we will have to store $2^k$ codewords at our worse, where $k$ is the dimension of the code. This is less than $2^{n-k}$ if $k<\frac{n}{2}$, where $n$ is the length. So for a code of dimension less than $\frac{n}{2}$, our algorithm takes less memory. Now we can view the weight wise sorted codewords as non-linear codes of constant weight. So we can store each sorted partition using the standard representation of a non-linear code using kernel of it and it's coset representatives as discussed in \cite{Villanueva2015}. Thus this will take shorter memory even. 


\section{Can any linear code be represented as a subcode of intertwining code?}
Let us consider an $l$-dimensional subspace of the $nk$-dimensional vector space over $\mathbb{F}_q$. This $nk$ dimensional vector space can be represented as the vector space of matrices $\mathbb{F}_q^{n \times k}$. Can we find a pair of matrices $A \in \mathbb{F}_q^{n \times n}$ and $C \in \mathbb{F}_q^{k \times k}$ to construct an intertwining code $\mathcal{R}(A,C)$ which contains the $l$ dimensional linear code? This problem can be formulated as follows.

Given a linear code $\mathcal{C}$ over $\mathbb{F}_q$ of length $nk$ and dimension $l$. So, the linear code $\mathcal{C}$ has a generator matrix $G$ of order $l \times nk$. Can we represent the linear code $\mathcal{C}$ as a subcode of an intertwining code $\mathcal{R}(A,C)$?

\subsection{Forming the equations and existence of solutions}
We represent each row of the generating matrix as an $n \times k$ order matrix. So, there are $l$ linearly independent matrices $B_1, B_2, \dots, B_l$ corresponding to the generator matrix of the known linear code $\mathcal{C}$. Our aim is to find two such non-zero matrices $A$ and $C$ which satisfy the equation $AB_i=B_iC$ for each $B_i$, $i=1, 2, \dots, l$. To find $A$ and $C$, let us consider the entries of $A$ and $C$ are variables. Then we get $n^2+k^2$ variables. For each matrix $B_i$, there are $nk$ equations and there will be a total $nkl$ equations satisfying these variables. Here we use the mapping  $\bar~:F_q^{n \times n} \rightarrow F_q^{n^2 \times 1}$ with $B \mapsto \bar{B}$, where the matrix $\bar{B}$ is formed by concatenating columns of $B$. Now the equation $AB_i =B_iC$ can be written as $$ AB_i - B_iC = O \Rightarrow \begin{bmatrix} I_n \otimes B_i^T & | & -B_i \otimes I_k \end{bmatrix} \begin{bmatrix} \bar{A} \\ \bar{C} \end{bmatrix}=O \Rightarrow D_i \begin{bmatrix} \bar{A} \\ \bar{C} \end{bmatrix} = O,$$ where $D_i= \begin{bmatrix} I_n \otimes B_i^T ~|-B_i \otimes I_k \end{bmatrix}$. Let $ D = \begin{bmatrix} D_1 & D_2 & \cdots & D_l \end{bmatrix}^T$. Then the above system of $l$ equations is written as 

\begin{equation}\label{eq_1}
 D \begin{bmatrix} \bar{A} \\ \bar{C} \end{bmatrix} =O.
\end{equation}
 Here $D_i$ is coming from each $B_i$. Now the final solution is the solution of the equation (\ref{eq_1}). The matrix $D$ is of order $nkl \times (n^2+k^2)$. Thus the existence of non-trivial solution of above equation is reached if $n^2+k^2 \geq nkl$. This is a sufficient condition.
 
Now we see that each $D_i$ consists of two blocks, i.e., $B_i \otimes I_k$ and another block $I_n \otimes B_i^T$. So $D_i= \begin{bmatrix} I_n \otimes B_i^T ~|-B_i \otimes I_k \end{bmatrix}=\begin{bmatrix} I_n \otimes B_i^T ~|~ O \end{bmatrix}+\begin{bmatrix} O ~| -B_i \otimes I_k \end{bmatrix}=A_1+A_2$.
 
Now $rank(D_i) \leq rank(A_1)+rank(A_2)$. Clearly, $rank(A_1)=n \cdot rank(B_i^T) = n \cdot rank(B_i)$ and $rank(A_2)= k \cdot rank(B_i)$. Now, we get $rank(D_i) \leq (n+k) \cdot rank(B_i)$. So, $rank(D) \leq \sum_{i=1}^l (n+k) \cdot rank(B_i)$. Now the final solution space will have dimension $ \geq n^2+k^2- \sum_{i=1}^l (n+k) \cdot rank(B_i)$. 

\section{Conclusion}
Centralizer codes, twisted centralizer codes and generalized twisted centralizer codes have length $n^2$ which is a reason that it cannot fit to most of the famous linear codes. But intertwining codes  can reach most of the linear codes because it is of length $nk$. So, we have taken intertwining codes and try to make a correspondence between it and existing linear codes. We have found an upper bound on minimum distance and proposed two decoding algorithms which take less storage memory.\\

\textbf{Acknowledgements}
 
The author Joydeb Pal is thankful to DST-INSPIRE for financial support to pursue his research work.


\end{document}